%% file: main.tex
\documentclass{article}
\usepackage[utf8]{inputenc}
\usepackage{amstext,amsmath,amssymb,amsgen,amsfonts,latexsym,theorem,pifont}

\usepackage{graphicx}
\usepackage{xcolor}
\usepackage{tikz}
\usetikzlibrary{arrows.meta,positioning,fit,calc}
\usepackage{float}

\usepackage{mytheorems}
\usepackage{mylayout}

\begin{document}
\sloppy
\title{
Classical Hardness of Learning Functions of Hamiltonians
}
\author{
  Sota Hashimoto\footnote{Graduate School of Engineering, Mie University, \texttt{426M527@m.mie-u.ac.jp}}
  \and
  Akinori Kawachi\footnote{Graduate School of Engineering, Mie University, \texttt{kawachi@info.mie-u.ac.jp}}
}
\date{}
\maketitle
\thispagestyle{plain}
\pagestyle{plain}

\begin{abstract}
\input{abstract}
\end{abstract}
\section{Introduction}\label{sec:intro}
\input{intro}

\section{Preliminaries}\label{sec:prelim}
\input{prelim}

\section{Main Result}\label{sec:main}
\input{result}
$\mbox{}$\\
\noindent{\bf Acknowledgements.}
This work is partially supported by 
JSPS Grant-in-Aid for Scientific Research (A) Nos.~21H04879, 23H00468, (C) No.~21K11887, JSPS Grant-in-Aid for Challenging Research (Pioneering) No.~23K17455, and MEXT Quantum Leap Flagship Program (MEXT Q-LEAP) Grant Number JPMXS0120319794. 

$\mbox{}$\\
\noindent{\bf AI Disclosure.}
During the preparation of this work, we used ChatGPT 6 Astra and
Claude Fable 5.1 to help refine technical ideas, support mathematical
analyses, and assist with editing and reviewing the \LaTeX{} source files.
This assistance concerned the development and refinement of technical
ideas and analyses, as well as the preparation and revision of the
manuscript source.
The authors take full responsibility for the correctness, originality,
and integrity of the manuscript, including its mathematical claims,
proofs, and references.

\bibliographystyle{alpha}
\bibliography{mybib}

\end{document}

%% file: abstract.tex
Morohoshi, Nakayama, Manabe, and Mitarai~\cite{Morohoshi_2025} proposed a physically motivated quantum machine learning problem in which the goal is to predict quantities of the form $\operatorname{Tr}[f(H)\rho]$ from classical descriptions of a Hamiltonian $H$ and a quantum state $\rho$, where $f$ is an unknown function. We call this problem Hamiltonian function learning in this paper. They constructed an efficient quantum learning algorithm under suitable conditions, while leaving a rigorous proof of average-case classical hardness open.
In this paper, we rigorously prove the average-case classical hardness for two distribution-specific Hamiltonian function learning problems for $f_{\cos,\pi}(\lambda)=\cos(\pi\lambda)$ and $f_{\exp,\beta}(\lambda)=e^{-\beta\lambda}$ discussed in~\cite{Morohoshi_2025} under the assumption of the average-case hardness of factoring random RSA moduli. More specifically, we show that an efficient classical randomized learner under squared loss whose output hypotheses are evaluable in classical polynomial time for either problem would yield a classical randomized polynomial-time algorithm for factoring random RSA moduli.

%% file: intro.tex
Quantum machine learning (QML) studies learning tasks involving quantum computation or quantum systems~\cite{Biamonte_2017,Havlivcek_2019,Schuld_2019,Caro_2022}. A central theoretical goal is to determine when quantum computation can provide a provable advantage over classical randomized polynomial-time algorithms. Whether such a separation exists depends on the precise learning model, including the input and data-access models, the performance measure, and the class of classical algorithms used for comparison.

Cryptographic assumptions have long been used to establish computational hardness in learning. Kearns and Valiant gave representation-independent hardness results for the probably approximately correct (PAC) learning by reducing cryptographic problems, including breaking RSA and factoring Blum integers, to learning problems~\cite{Kearns_1994}, and Servedio and Gortler obtained quantum--classical computational separations for PAC learning under the hardness of factoring Blum integers and for exact learning from membership queries under the existence of one-way functions~\cite{Servedio_2004}.  More recently, Liu, Arunachalam, and Temme~\cite{Liu_2021} established such a separation based on the hardness of the discrete logarithm problem, and Gyurik and Dunjko~\cite{Gyurik_2023} examined how the learning model determines where the classical difficulty arises. These computational separations should be distinguished from sample-complexity advantages, which are known to be limited: quantum and classical sample complexities coincide up to constant factors in the PAC and agnostic models~\cite{Arunachalam_2018}, and, for predicting outcomes of quantum experiments with small average error, classical learners need only a number of experiments comparable to that of quantum learners~\cite{Huang_2021}.

Morohoshi, Nakayama, Manabe, and Mitarai~\cite{Morohoshi_2025} introduced the learning problem studied in this paper, aiming at a QML problem of practical relevance. In contrast to cryptographically constructed learning problems, their problem is motivated by a physical setting in which the Hamiltonian of a quantum system is known, but the relation between the Hamiltonian and experimental outputs is not. They model such an output as a quantity of the form $\operatorname{Tr}[f(H)\rho]$, where the unknown function $f$ represents this relation, and the goal is to predict this quantity from a Hamiltonian $H$ and a quantum state $\rho$, given labeled examples.
We call this problem Hamiltonian function learning in this paper. They constructed an efficient quantum learning algorithm under suitable conditions on the unknown function $f$, the input Hamiltonian $H$, and the input quantum state $\rho$. Furthermore, they verified the feasibility of their proposed algorithm by implementing it for $f(\lambda)=e^{-\beta\lambda}$ on a superconducting device using up to 40 qubits. 

They also discussed the classical hardness of their learning task as evidence for a quantum advantage in QML.
For $f(\lambda)=\cos(t\lambda)$, they connected $\operatorname{Tr}[\cos(tH)\rho]=\operatorname{Re}\operatorname{Tr}[e^{-itH}\rho]$ to Hamiltonian simulation and argued that predicting this quantity should be classically hard for $t=\operatorname{poly}(n)$. 

However, they left a rigorous proof of average-case classical hardness for Hamiltonian function learning open. In this work, we present such a result under the assumption that factoring random RSA moduli is classically hard on average. We construct a distribution-specific $f_{\cos,\pi}$-Hamiltonian function learning problem and, for every fixed rational constant $\beta>0$, a distribution-specific $f_{\exp,\beta}$-Hamiltonian function learning problem, both of which are efficiently learnable by their quantum algorithm. Under the above assumption, neither problem is efficiently learnable by a classical randomized algorithm under squared loss when the output hypotheses are required to be classically evaluable. In particular, the hardness holds even though $f$ is fixed in advance and known to the learner. The Hamiltonians in our hard instances are projectors specified succinctly by polynomial-size quantum circuits rather than local Hamiltonians.

Related quantum--classical learning separations have been established in several settings. P\'erez-Guijarro, Pag\`es-Zamora, and Fonollosa studied conditions relating quantum learning speedups to quantum computational speedups and applied their results to learning problems based on factorization~\cite{PerezGuijarro_2024}. Yamasaki, Isogai, and Murao provided a framework for deriving supervised-learning advantages from general quantum computational advantages and constructed corresponding data-preparation protocols~\cite{Yamasaki_2023}. Molteni, Gyurik, and Dunjko proved quantum--classical separations for learning unknown quantum observables in a distribution-free PAC model under the assumption $\mathsf{BQP}\not\subseteq\mathsf{P/poly}$~\cite{Molteni_2025}. Barthe, Yaghubi Rad, Grossi, and Dunjko studied supervised learning of unknown Hamiltonian dynamics and proved quantum--classical learning separations under complexity-theoretic assumptions~\cite{Barthe_2025}.

Our result differs from these works in several respects. P\'erez-Guijarro et al.~require pointwise accuracy on every input in the support, whereas our hardness result holds even for learners that are only required to achieve small expected squared loss. 
Moreover, their explicit separation examples are based on arithmetic functions of prime factors, whereas we establish hardness for the physically motivated problem introduced by \cite{Morohoshi_2025}.
Yamasaki et al.~give a more general framework for deriving learning advantages from quantum computational advantages, assuming hardness against classical polynomial-time algorithms even with polynomial-length randomized advice, as captured by $\mathsf{HeurFBPP/rpoly}$. 
In our setting, the reduction requires only average-case hardness against uniform classical randomized polynomial-time algorithms. All training examples in our construction can also be generated in classical polynomial time.
Moreover, the reduction is not actually specific to factoring: it extends to problems with a polynomial-length witness uniquely determined by each instance, for which instance--witness pairs are classically samplable, the witness can be recovered in quantum polynomial time and verified in classical polynomial time, and its recovery is hard on average for classical randomized polynomial-time algorithms; see Remark~\ref{rem:generalization}. Molteni et al.~obtain a distribution-free separation under the different assumption $\mathsf{BQP}\not\subseteq\mathsf{P/poly}$, whereas our hardness result holds for expected squared loss under a single explicit, efficiently samplable input distribution. Barthe et al.~study the different problem of learning unknown Hamiltonian dynamics, whereas our result concerns the Hamiltonian function learning problem introduced by Morohoshi et al. Consequently, our result gives rigorous average-case classical hardness even when $f$ is fixed in advance; for $f_{\cos,\pi}$, it applies to the fixed choice $t=\pi$, and for $f_{\exp,\beta}$, it includes the choice of $f$ considered in the experiment on quantum hardware of Morohoshi et al.

Our proof follows the cryptographic reduction strategy of Kearns and Valiant~\cite{Kearns_1994}. To outline the reduction, let $\widetilde{\mathcal D}_n$ be the efficiently samplable joint distribution over triples $(N,p,q)$ generated by a fixed RSA modulus generator~\cite{Rivest_1978}, where $N=pq$, and let $\mathcal D_n$ be its marginal distribution on $N$. The distinction between the two distributions is essential for the reduction: independent samples from $\widetilde{\mathcal D}_n$ provide the factors needed to generate training labels, whereas a factoring challenge provides only $N\sim\mathcal D_n$. We assume that, given such a random challenge modulus, no classical randomized polynomial-time algorithm can output a nontrivial factorization with non-negligible probability.

Our construction uses a verification circuit based on Shor's factoring algorithm~\cite{Shor_1997}. Given an RSA modulus $N$, a bit position $j$, and a bit $b\in\{0,1\}$, we construct a polynomial-size quantum circuit $U_{N,j,b}$ whose acceptance probability is overwhelmingly close to one if $b$ equals the $j$-th bit of $\min\{p,q\}$ and overwhelmingly close to zero otherwise. From this circuit, we define the non-local projector Hamiltonian $H_{N,j,b}:=U_{N,j,b}^{\dagger}\Pi_{\mathrm{acc}}U_{N,j,b}$, where $\Pi_{\mathrm{acc}}$ is the acceptance projector, and use the all-zero initial state $\rho_0$.

Since $H_{N,j,b}$ is a projector, both $\operatorname{Tr}[\cos(\pi H_{N,j,b})\rho_0]$ and $\operatorname{Tr}[e^{-\beta H_{N,j,b}}\rho_0]$ are affine functions of the acceptance probability of $U_{N,j,b}$. In particular, the target values for $b=0$ and $b=1$ are separated by a constant: when $b$ equals the $j$-th bit of $\min\{p,q\}$, the two target values are close to $-1$ and $e^{-\beta}$, respectively, whereas both are close to $1$ otherwise.

Now suppose that an efficient classical learner exists for either of the two constructed problems. Given a challenge modulus $N\sim\mathcal D_n$, a classical factoring algorithm generates an i.i.d.~training sample from independently drawn triples $(N_i,p_i,q_i)\sim\widetilde{\mathcal D}_n$ and runs the learner to obtain a hypothesis $g$. It then compares $g(x_{N,j,0})$ and $g(x_{N,j,1})$ for each bit position $j$. Sufficiently small inverse polynomial squared loss ensures that, with probability at least $2/3$, these comparisons recover all bits of $\min\{p,q\}$ and hence a factor of $N$. Figure~\ref{fig:classical-hardness-reduction} summarizes the reduction.

\begin{figure}[H]
\centering
\begin{tikzpicture}[
  >=Latex,
  font=\small,
  trainbox/.style={
    draw,
    rounded corners=2pt,
    align=center,
    text width=25mm,
    minimum height=16mm,
    inner sep=3pt
  },
  learnerbox/.style={
    draw,
    rounded corners=2pt,
    align=center,
    text width=25mm,
    minimum height=16mm,
    inner sep=3pt
  },
  querybox/.style={
    draw,
    rounded corners=2pt,
    align=center,
    text width=25mm,
    minimum height=16mm,
    inner sep=3pt
  },
  recoverbox/.style={
    draw,
    rounded corners=2pt,
    align=center,
    text width=25mm,
    minimum height=16mm,
    inner sep=3pt
  },
  io/.style={
    align=center,
    text width=17mm,
    font=\small
  },
  arr/.style={
    -Latex,
    semithick
  }
]

\node[trainbox] (train) {
  Generate sample $S$\\
  from i.i.d.~triples\\
  $(N_i,p_i,q_i)\sim\widetilde {\mathcal {D}}_n$
};

\node[learnerbox, right=3.5mm of train] (learner) {
  Run learner $\mathcal {A}$\\
  on $(1^n,\varepsilon(n),\delta(n),S)$\\
  to obtain $g$
};

\node[querybox, right=3.5mm of learner] (query) {
  For each $j$, construct\\
  $x_{N,j,0},x_{N,j,1}$\\
  and evaluate $g$
};

\node[recoverbox, right=3.5mm of query] (recover) {
  Recover all $\widehat z_j$,\\
  reconstruct and verify\\
  the factor $\widehat p$
};

\coordinate (routelevel) at ([yshift=-2mm]train.south);

\node[
  draw,
  rounded corners=3pt,
  fit=(train)(learner)(query)(recover)(routelevel),
  inner sep=4.5pt,
  label={[font=\small\bfseries]above:
    Factoring algorithm $\mathcal {B}$}
] (B) {};

\node[io, anchor=east] (input)
  at ($(B.west |- routelevel)+(-3mm,0)$) {
  challenge\\
  modulus $N$
};

\node[io, anchor=west] (output)
  at ($(B.east |- recover.east)+(3mm,0)$) {
  factors\\
  of $N$
};

\draw[arr] (train.east) -- (learner.west);
\draw[arr] (learner.east) -- (query.west);
\draw[arr] (query.east) -- (recover.west);

\draw[arr]
  (input.east)
  -- (B.west |- routelevel)
  -- (query.south |- routelevel)
  -- (query.south);

\draw[arr] (recover.east) -- (output.west);

\end{tikzpicture}

\caption{
Reduction from an efficient learner $\mathcal {A}$ to a
classical factoring algorithm $\mathcal {B}$.
}
\label{fig:classical-hardness-reduction}
\end{figure}
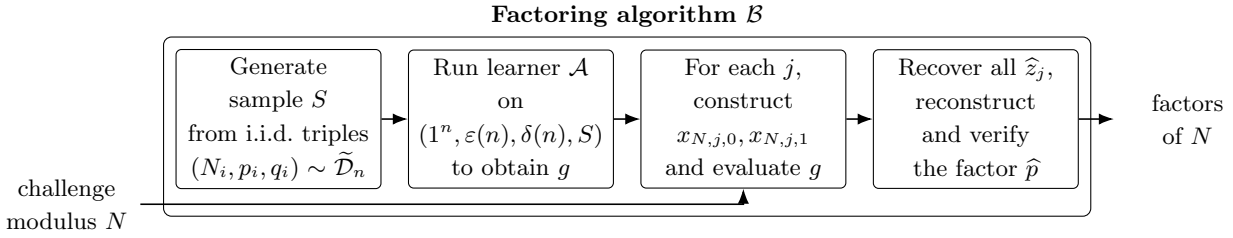

%% file: prelim.tex
We first define basic notions and notation used throughout this paper.
For a positive integer $n$, let $\operatorname{poly}(n)$ denote an unspecified polynomial in $n$. 
For $L\in\mathbb{N}$, we write $[L]:=\{1,2,\ldots,L\}$. For an integer $z$ with $0\le z<2^L$ and $j\in[L]$, let $\operatorname{bit}^{(L)}_j(z)$ denote the $j$-th bit of the $L$-bit binary representation of $z$, padded with leading zeros if necessary. We count bit positions from the least significant bit. Hence, $j=1$ denotes the least significant bit. When $L$ is clear from context, we write $\operatorname{bit}_j(z)$.

For an operator $X$, let $\operatorname{spec}(X)$ denote its spectrum.
For a function $f: \mathbb{R}\to\mathbb{R}$ and a Hermitian operator $H$ with spectral decomposition $H=\sum_{j}\lambda_j\,\Pi_j$, where $\lambda_j$ are the distinct eigenvalues of $H$ and $\Pi_j$ are the corresponding orthogonal projectors, we define
$f(H):=\sum_{j} f(\lambda_j)\,\Pi_j$. For instance, $f(\lambda)=\cos(t\lambda)$ yields $f(H)=\cos(tH)=\operatorname{Re}\,e^{-itH}$, and $f(\lambda)=e^{-\beta\lambda}$ yields the (unnormalized) Gibbs operator $e^{-\beta H}$.

A function $\nu:\mathbb{N}\to\mathbb{R}_{\ge 0}$ is negligible if, for every positive polynomial $p$, there exists $n_0$ such that $\nu(n)\le 1/p(n)$ for all $n\ge n_0$. We write $\operatorname{negl}(n)$ for an unspecified negligible function. A classical randomized polynomial-time algorithm means a uniform probabilistic algorithm running in time polynomial in the input length. A quantum polynomial-time algorithm means a uniform family of quantum circuits of polynomial size. Unless otherwise specified, all algorithms are uniform, and all probabilities are taken over both the randomness of the relevant distributions and the internal randomness of the algorithms.

We next formalize the regression problems considered in this paper. A learner receives labeled training examples $(x,y)$ drawn from a fixed distribution $\mathcal P$ and outputs a hypothesis $g$ that approximates a real-valued target function $\phi$ under squared loss.

\begin{definition}[Distribution-specific regression problem]\label{def:Distribution-specific regression problem}

A distribution-specific regression problem is a family $\mathcal{R}={(\mathcal X_n,\phi_n,\mathcal{P}_n)}_{n\in\mathbb{N}}$, where, for each $n$, $\mathcal X_n\subseteq \{0,1\}^{\ell(n)}$ is a finite instance space ($\ell(n)=\operatorname{poly}(n)$), $\phi_n:\mathcal X_n\to[-1,1]$ is a target function, and $\mathcal{P}_n$ is a distribution over $\mathcal X_n\times \mathcal Y_n$, for some finite $\mathcal Y_n\subseteq[-1,1]\cap\mathbb{Q}$ whose elements have bit length polynomial in $n$. Let $\mu_n:=\mathcal{P}_n|_{\mathcal X_n}$ denote the marginal of $\mathcal{P}_n$ on $\mathcal X_n$.
For a hypothesis $g:\mathcal X_n\to\mathbb{R}$, define the expected squared loss of the hypothesis as 
\[
R_{\mu_n}^{\phi_n}(g):=\mathbb{E}_{x\sim\mu_n}\left[\bigl(g(x)-\phi_n(x)\bigr)^2\right].
\]
The label $y$ in an example $(x,y)\sim \mathcal{P}_n$ need not equal $\phi_n(x)$.
For brevity, we write $R_{\mu_n}(g):=R_{\mu_n}^{\phi_n}(g)$ if the target function is clear from the context.

\end{definition}

We formulate the $f$-Hamiltonian function learning problems studied
in this paper as distribution-specific regression problems as follows.

\begin{definition}[$f$-Hamiltonian function learning problem]
\label{def:Hamiltonian function learning instance}
Let $f:[0,1]\to[-1,1]$ be a function. An $f$-Hamiltonian function learning ($f$-HFL) problem is a distribution-specific regression problem $\mathcal{L}=(\mathcal{X}_n,F_n,\mathcal P_n)_{n\in\mathbb{N}}$ such that each $x\in\mathcal{X}_n$ is a polynomial-length classical description of a pair of a Hamiltonian $H_x$ and a density operator $\rho_x$ on $m(n)=\operatorname{poly}(n)$ qubits, with
$\operatorname{spec}(H_x)\subseteq[0,1]$. The target function is 
$F_n(x):=\operatorname{Tr}[f(H_x)\rho_x]$. 
\end{definition}
For $t\in\mathbb{R}$ and a rational $\beta>0$, define $f_{\cos,t}(\lambda):=\cos(t\lambda)$ and $f_{\exp,\beta}(\lambda):=e^{-\beta\lambda}$.
We refer to the corresponding $f$-HFL problems as
$f_{\cos,t}$-HFL and $f_{\exp,\beta}$-HFL problems, respectively.
Their target functions are $F_n^{\cos,t}(x):=\operatorname{Tr}[\cos(tH_x)\rho_x]$ and $F_n^{\exp,\beta}(x):=\operatorname{Tr}[e^{-\beta H_x}\rho_x]$, respectively. In the two problems constructed in Section~\ref{sec:main}, the training labels approximate the corresponding target values up to negligible error.

To formulate our classical hardness result, we define efficient classical learnability under squared loss.

\begin{definition}[Efficient classical learnability under squared loss]\label{def:Efficient classical learner under squared loss}

Let $\mathcal{L}={(\mathcal X_n,F_n,\mathcal{P}_n)}_{n\in\mathbb{N}}$ be an $f$-HFL problem. We say that $\mathcal{L}$ is efficiently learnable by a classical randomized algorithm under squared loss if there exist a polynomial $M$ and a classical randomized polynomial-time algorithm $\mathcal{A}$ such that the following holds. For every $n\in\mathbb{N}$ and rational $\varepsilon,\delta\in(0,1)$, the learner $\mathcal{A}$, on input $1^n$, $\varepsilon$, $\delta$, and an i.i.d.~training sample
\[
S=\{(x_i,y_i)\}_{i=1}^{M(n,\varepsilon^{-1},\delta^{-1})}
\sim
\mathcal{P}_n^{\otimes M(n,\varepsilon^{-1},\delta^{-1})},
\]
outputs a description $s$ of a hypothesis $g_s:\mathcal X_n\to[-1,1]\cap\mathbb{Q}$ such that the mapping $(1^n,s,x)\mapsto g_s(x)$ is computable by a single deterministic polynomial-time algorithm. With probability at least $1-\delta$ over $S$ and the randomness of $\mathcal{A}$, $R_{\mu_n}(g_s)\le \varepsilon$.

\end{definition}

The polynomial-time evaluation requirement is essential for our reduction, which evaluates the learned hypothesis on inputs constructed from the challenge modulus. 

Our classical hardness results rely on the average-case hardness of integer factorization. We therefore introduce the RSA modulus distribution and the corresponding hardness assumption.

\begin{definition}[RSA modulus distribution]\label{def:RSA modulus distribution}

Let $\operatorname{GenRSA}$ be a classical randomized polynomial-time RSA modulus generator that outputs a triple $(N,p,q)$ on input $1^n$, where $p$ and $q$ are distinct odd primes of $\Theta(n)$ bits and $N=pq$. We denote by $\widetilde{\mathcal{D}}_n$ the joint distribution of $(N,p,q)$ generated by $\operatorname{GenRSA}(1^n)$. We denote by $\mathcal{D}_n$ the marginal distribution of $N$. We also fix a polynomial-time computable, polynomially bounded function $L:\mathbb{N}\to\mathbb{N}$ such that $\min\{p,q\}<2^{L(n)}$ for every $(N,p,q)\in\operatorname{supp}(\widetilde{\mathcal{D}}_n)$.

\end{definition}
When referring to the bits of $\min\{p,q\}$, we use its $L(n)$-bit binary representation padded with leading zeros.

For simplicity, we idealize $\operatorname{GenRSA}$ as always returning a valid triple. 
A bounded-time implementation may instead output a failure symbol with negligible probability. 
Including this failure event explicitly changes the probability bounds below only by a negligible additive term.

The distinction between $\widetilde{\mathcal D}_n$ and $\mathcal D_n$ is important for the reduction. In standard RSA modulus generation~\cite{Rivest_1978}, the factors are generated first and the modulus $N=pq$ is then computed efficiently; primality of candidate factors can be tested in polynomial time, for example, by the AKS primality test~\cite{Agrawal_2004}. Thus, a sample generator drawing $(N,p,q)\sim\widetilde{\mathcal D}_n$ obtains the factors and can use them to generate training labels. In the average-case factoring problem, however, the algorithm receives only $N\sim \mathcal D_n$. Efficient joint generation of $(N,p,q)$ does not provide an efficient procedure for recovering $p$ and $q$ from a given challenge modulus $N$.

We formalize the average-case hardness assumption as follows.

\begin{assumption}[Average-case factoring hardness]\label{ass:Average-case hardness}

For every classical randomized polynomial-time algorithm $\mathcal{B}$, there exists a negligible function $\nu$ such that, for all sufficiently large $n$,
\[
\Pr_{\substack{(N,p,q)\sim\widetilde{\mathcal{D}}_n,\, \mathcal{B}}}
\left[
\exists\,p',q'\in\mathbb{N}:
\mathcal B(1^n,N)=(p',q'),\
1<p',q'<N,\
p'q'=N
\right]
\le \nu(n).
\]
Equivalently, given only a random RSA modulus $N\sim \mathcal{D}_n$, no classical randomized polynomial-time algorithm outputs a nontrivial factorization of $N$ with non-negligible probability.

\end{assumption}

%% file: result.tex
Under Assumption~\ref{ass:Average-case hardness}, we establish the following average-case classical hardness result for two explicit $f$-HFL problems.

\begin{theorem}[Main theorem]
\label{thm:main}
For every fixed rational constant $\beta>0$, there exist an $f_{\cos,\pi}$-HFL problem $\mathcal L^{\cos,\pi}$ and an $f_{\exp,\beta}$-HFL problem $\mathcal L^{\exp,\beta}$ such that, if either problem is efficiently learnable by a classical randomized algorithm under squared loss in the sense of Definition~\ref{def:Efficient classical learner under squared loss}, then there exists a classical randomized polynomial-time algorithm $\mathcal B$ such that, for all sufficiently large $n$,
\[
    \Pr_{\substack{(N,p,q)\sim\widetilde{\mathcal{D}}_n,\, \mathcal B}}
    \left[
    \exists\,p',q'\in\mathbb{N}:
    \mathcal{B}(1^n,N)=(p',q'),\
    1<p',q'<N,\
    p'q'=N
    \right]
    \ge \frac{2}{3}.
\]
Consequently, under Assumption~\ref{ass:Average-case hardness}, neither problem is efficiently learnable by a classical randomized algorithm under squared loss in the sense of Definition~\ref{def:Efficient classical learner under squared loss}.
\end{theorem}

\begin{proof}
We construct an $f_{\cos,\pi}$-HFL problem $\mathcal L^{\cos,\pi}$ and an $f_{\exp,\beta}$-HFL problem $\mathcal L^{\exp,\beta}$, and show that an efficient classical learner for either problem yields a classical randomized polynomial-time algorithm $\mathcal B$ for factoring a random modulus drawn from $\mathcal D_n$. 
Given a challenge modulus $N$, the algorithm $\mathcal B$ generates an independent training sample using independent RSA triples and runs the learner to obtain a hypothesis $g$. For each bit position $j$, it then compares $g$ on two inputs $x_{N,j,0}$ and $x_{N,j,1}$ constructed from $N$. The construction below ensures that these comparisons recover the bits of $\min\{p,q\}$ whenever $g$ is sufficiently accurate, after which $\mathcal B$ reconstructs and verifies the factor.
For this reduction, we need the learning problems to satisfy three properties: training samples can be generated efficiently from independently generated RSA triples, each challenge input $x_{N,j,b}$ can be constructed from $N,j,b$ alone without knowing the factors of $N$, and the target values for the two choices of $b$ are sufficiently separated so that their ordering determines the corresponding bit of $\min\{p,q\}$. 

We now construct two $f$-HFL problems with these properties.
Fix $n\in\mathbb{N}$. Draw $(N,p,q)\sim\widetilde{\mathcal{D}}_n$, and choose a bit position $j\in[L(n)]$ and a bit $b\in\{0,1\}$ independently and uniformly at random. Let 
\begin{equation}
z_j:=\operatorname{bit}^{(L(n))}_j(\min\{p,q\}).
\label{eq:estimation-bit}
\end{equation}
From $N,j,b$, we construct a polynomial-size quantum circuit $U_{N,j,b}$ based on Shor's factoring algorithm~\cite{Shor_1997} with success-probability amplification. Any classical random bits used in the factoring procedure are generated coherently using ancillary qubits and retained as part of the workspace, so that $U_{N,j,b}$ is a fixed unitary circuit determined by $N,j,b$. The circuit coherently performs the factoring procedure with designated registers containing candidate factors $(p',q')$ and a designated one-qubit output register $o$ that indicates acceptance or rejection. The classical post-processing and subsequent checks are implemented reversibly. For each repetition, the circuit verifies that $1<p',q'<N$ and $p'q'=N$. If the verification succeeds, it computes $z:=\min\{p',q'\}$ and checks whether $b=\operatorname{bit}^{(L(n))}_j(z)$. The circuit accepts if this condition is satisfied in at least one repetition, and rejects otherwise. Since $N=pq$ with distinct primes $p$ and $q$, every valid nontrivial factorization of $N$ yields the same value $\min\{p,q\}$. By the correctness of Shor's algorithm and the success-probability amplification above, the acceptance probability of $U_{N,j,b}$ is therefore close to one when $b=\operatorname{bit}^{(L(n))}_j(\min\{p,q\})$, and close to zero otherwise.

From $U_{N,j,b}$, we define the projector Hamiltonian 
\begin{equation}
H_{N,j,b}:=U_{N,j,b}^{\dagger}\Pi_{\mathrm{acc}}U_{N,j,b},
\label{eq:projector-hamiltonian}
\end{equation}
where $\Pi_{\mathrm{acc}}$ is the acceptance projector on the designated one-qubit output register $o$. Since $\Pi_{\mathrm{acc}}$ is a projector and $U_{N,j,b}$ is unitary, $H_{N,j,b}$ is also a projector. In particular, we have $H_{N,j,b}^2=H_{N,j,b}$ and $\operatorname{spec}(H_{N,j,b})\subseteq\{0,1\}$. Let $\rho_0:=|0\cdots0\rangle\langle0\cdots0|$ be the all-zero initial state on the registers of $U_{N,j,b}$. 
For each $N\in\operatorname{supp}(\mathcal D_n)$, $j\in[L(n)]$, and $b\in\{0,1\}$, let $x_{N,j,b}$ denote a classical description of the pair $(H_{N,j,b},\rho_0)$, where $H_{N,j,b}$ is represented succinctly by the classical description of the quantum circuit $U_{N,j,b}$ and the acceptance projector $\Pi_{\mathrm{acc}}$.

We define the common instance space by
\[
    \mathcal{X}_n:=
    \left\{
        x_{N,j,b}:N\in\operatorname{supp}(\mathcal{D}_n),\ j\in[L(n)],\ b\in\{0,1\}
    \right\}.
\]
Since $U_{N,j,b}$ has polynomial size, each $x_{N,j,b}$ has polynomial length. We use a canonical encoding padded to a common polynomial length $\ell(n)$, so that $\mathcal{X}_n\subseteq\{0,1\}^{\ell(n)}$. We also pad the circuit registers with unused ancilla qubits so that, for each $n$, all $H_{N,j,b}$ and $\rho_0$ act on a common number $m(n)=\operatorname{poly}(n)$ of qubits.

Define $y^{\cos,\pi}:=1-2\mathbf{1}[b=z_j]$. Let $c_\beta:=e^{-\beta}$. Since $\beta>0$ is a fixed rational constant, for each $n$ we fix a rational number $\widetilde c_{\beta,n}\in[0,1]$ of bit length polynomial in $n$ that can be computed in classical polynomial time and satisfies
$|\widetilde c_{\beta,n}-c_\beta|\le 2^{-n}$. We define $y^{\exp,\beta}:=1-\bigl(1-\widetilde c_{\beta,n}\bigr)\mathbf{1}[b=z_j]$. We denote by $\mathcal P^{\cos,\pi}_n$ the distribution of $(x_{N,j,b},y^{\cos,\pi})$ generated by the procedure above, and by $\mathcal P^{\exp,\beta}_n$ the corresponding distribution of $(x_{N,j,b},y^{\exp,\beta})$. By construction, the two distributions have the same marginal distribution $\mu_n$ on $\mathcal{X}_n$. 
Define the target functions on $\mathcal X_n$ by $F_n^{\cos,\pi}(x):=\operatorname{Tr}[\cos(\pi H_x)\rho_x]$, $F_n^{\exp,\beta}(x):=\operatorname{Tr}[e^{-\beta H_x}\rho_x]$. We then set $\mathcal L^{\cos,\pi}:=(\mathcal X_n,F_n^{\cos,\pi},\mathcal P_n^{\cos,\pi})_{n\in\mathbb N}$, $\mathcal L^{\exp,\beta}:=(\mathcal X_n,F_n^{\exp,\beta},\mathcal P_n^{\exp,\beta})_{n\in\mathbb N}$. 
By construction, these are $f_{\cos,\pi}$-HFL and $f_{\exp,\beta}$-HFL problems, respectively, in the sense of Definition~\ref{def:Hamiltonian function learning instance}.

We next verify that the constructed problems have the efficiency and target-separation properties needed for the reduction.
\begin{lemma}[Efficient sample generation and target-value bounds]
\label{lem:sample-target}
Both $\mathcal P^{\cos,\pi}_n$ and $\mathcal P^{\exp,\beta}_n$ are efficiently samplable by classical randomized polynomial-time algorithms. Moreover, for every $(N,p,q)\in\operatorname{supp}(\widetilde{\mathcal{D}}_n)$, every $j\in[L(n)]$ and every $b\in\{0,1\}$, the input $x_{N,j,b}$ can be constructed in classical polynomial time from $N,j,b$ alone, without using $p$ or $q$. There exists a negligible function $\nu$ such that, for every such
triple, $j$, and $b$, writing $z_j:=\operatorname{bit}^{(L(n))}_j(\min\{p,q\})$, we have
\[
    \left|
        F^{\cos,\pi}_n(x_{N,j,b})-\bigl(1-2\mathbf{1}[b=z_j]\bigr)
    \right|
    \leq 2\nu(n),
\]
and
\[
    \left|
        F^{\exp,\beta}_n(x_{N,j,b})-
        \left(
            1-(1-e^{-\beta})\mathbf{1}[b=z_j]
        \right)
    \right|
    \leq
    (1-e^{-\beta})\nu(n).
\]
Moreover, every $(x,y)\in\operatorname{supp}(\mathcal P^{\cos,\pi}_n)$ satisfies $|y-F^{\cos,\pi}_n(x)|\le 2\nu(n)$, and every $(x,y)\in\operatorname{supp}(\mathcal P^{\exp,\beta}_n)$ satisfies $|y-F^{\exp,\beta}_n(x)|\le \nu(n)+2^{-n}$.
\end{lemma}
\begin{proof}
Define $a_{N,j,b}:=\Pr[U_{N,j,b}\text{ accepts on }\rho_0]$.
By construction, $U_{N,j,b}$ uses an amplified unitary implementation of Shor's factoring algorithm~\cite{Shor_1997}. Since $p$ and $q$ are distinct odd primes, the modulus $N=pq$ is an odd composite that is not a prime power. Intermediate measurements can be removed with polynomial overhead~\cite{aharonov_1998}, and the classical post-processing, factor verification, and comparison with $b$ can be implemented reversibly~\cite{bennett_1973}.

By the correctness of Shor's factoring algorithm and standard amplification, the number of repetitions can be chosen polynomially in $n$ so that, for every $(N,p,q)\in\operatorname{supp}(\widetilde{\mathcal D}_n)$, the amplified factoring procedure produces a valid nontrivial factorization of $N$ in the candidate registers with probability at least $1-\nu(n)$, where $\nu$ is negligible. Moreover, the reversible verification step ensures that any accepted
candidate factorization is a valid nontrivial factorization of $N$.

Since every valid nontrivial factorization of $N=pq$ yields $\min\{p,q\}$, we obtain
\[
    \left|
        a_{N,j,b}
        -
        \mathbf{1}[b=z_j]
    \right|
    \leq
    \nu(n).
\]
By Eq.~\eqref{eq:projector-hamiltonian}, we have $a_{N,j,b}=\operatorname{Tr}[H_{N,j,b}\rho_0]$.
Since $H_{N,j,b}$ is a projector, $\cos(\pi H_{N,j,b})=I-2H_{N,j,b}$. Therefore,
\[
    F_n^{\cos,\pi}(x_{N,j,b})
    =
    \operatorname{Tr}
    [\cos(\pi H_{N,j,b})\rho_0] 
    =
    1-2a_{N,j,b}.
\]
It follows that
\[
    \left|
        F_n^{\cos,\pi}(x_{N,j,b})
        -
        \bigl(1-2\mathbf{1}[b=z_j]\bigr)
    \right|
    =
    2\left|
        a_{N,j,b}-\mathbf{1}[b=z_j]
    \right|
    \le 2\nu(n).
\]
For $\mathcal L^{\exp,\beta}$, since $H_{N,j,b}$ is a projector, $e^{-\beta H_{N,j,b}}=I-(1-e^{-\beta})H_{N,j,b}$. Consequently,
\[
    F_n^{\exp,\beta}(x_{N,j,b})
    =
    \operatorname{Tr}
    [e^{-\beta H_{N,j,b}}\rho_0]
    =
    1-(1-e^{-\beta})a_{N,j,b}.
\]
It follows that
\[
    \left|
        F_n^{\exp,\beta}(x_{N,j,b})
        -
        \bigl(
            1-(1-e^{-\beta})\mathbf{1}[b=z_j]
        \bigr)
    \right|
    =
    (1-e^{-\beta})
    \left|
        a_{N,j,b}-\mathbf{1}[b=z_j]
    \right|
    \le
    (1-e^{-\beta})\nu(n).
\]

For the labels, $y^{\cos,\pi}=1-2\mathbf{1}[b=z_j]$, so the first bound immediately gives
\[
    |y^{\cos,\pi}-F^{\cos,\pi}_n(x_{N,j,b})|
    \le 2\nu(n).
\]
For $y^{\exp,\beta}$, using
$|\,\widetilde c_{\beta,n}-e^{-\beta}\,|\le 2^{-n}$ together with the second bound gives
\[
    |y^{\exp,\beta}-F^{\exp,\beta}_n(x_{N,j,b})|
    \le \nu(n)+2^{-n}.
\]

Finally, we verify the remaining efficiency claims. By Definition~\ref{def:RSA modulus distribution}, $(N,p,q)$ can be sampled in classical randomized polynomial time using $\operatorname{GenRSA}(1^n)$. Since $L(n)$ is polynomial-time computable and polynomially bounded, $j\in[L(n)]$ and $b\in\{0,1\}$ can also be sampled in classical randomized polynomial time. The circuit description of $U_{N,j,b}$, and hence $x_{N,j,b}$, can be generated in classical polynomial time from $N,j,b$ alone, without using $p$ or $q$. The label $y^{\cos,\pi}$ is computable from the generated factors $p,q$, and the label $y^{\exp,\beta}$ is also computable in classical polynomial time using $\widetilde c_{\beta,n}$.
\end{proof}

We now use Lemma~\ref{lem:sample-target} to complete the reduction.
Fix one of the two constructed problems $\mathcal L=(\mathcal X_n,F_n,\mathcal P_n)_{n\in\mathbb N}$, suppose that it is efficiently learnable, and let $\mathcal A$ be a corresponding learner. Set $\gamma_\beta:=\frac{\beta}{1+\beta}$, $\tau(n):=\frac{\gamma_\beta^2}{128L(n)}$, $\varepsilon(n):=\frac{\tau(n)}{200}$ and $\delta(n):=\frac{1}{200}$. Since $\beta>0$ is fixed, $\varepsilon(n)$ is inverse polynomial. On input $1^n$ and a challenge modulus $N\sim\mathcal D_n$, $\mathcal B$ proceeds as follows.

\begin{enumerate}
\item Independently of the challenge modulus $N$, generate an i.i.d.~training sample
\[
    S=
    \{(x_i,y_i)\}_{i=1}^{M(n,\varepsilon(n)^{-1},\delta(n)^{-1})}
    \sim
    \mathcal P_n^{\otimes M(n,\varepsilon(n)^{-1},\delta(n)^{-1})}
\]
using the sampler from Lemma~\ref{lem:sample-target}.
\item Run the learner $\mathcal A$ on input $(1^n,\varepsilon(n),\delta(n),S)$. Let $s$ be the resulting hypothesis description and define $g:=g_s$.
\item For every bit position $j\in[L(n)]$, construct the two inputs $x_{N,j,0}$ and $x_{N,j,1}$, and define 
\[
\widehat z_j(g)
    :=
    \begin{cases}
        1, & \text{if } g(x_{N,j,1})\le g(x_{N,j,0}),\\
        0, & \text{otherwise}.
    \end{cases}
\]
For brevity, write $\widehat z_j:=\widehat z_j(g)$.
\item Reconstruct
\[
    \widehat p:=\sum_{j=1}^{L(n)} 2^{j-1}\widehat z_j.
\]
\item Check whether $1<\widehat p<N$ and $\widehat p\mid N$. If the check succeeds, output $\left(\widehat p,\frac{N}{\widehat p}\right)$. Otherwise, output failure.
\end{enumerate}
The algorithm $\mathcal B$ runs in classical randomized polynomial time. Since $\varepsilon(n)^{-1}=\operatorname{poly}(n)$ and $\delta(n)^{-1}=200$, the sample size in Step~1 is polynomial in $n$. Each sample and each challenge input can be generated in polynomial time by Lemma~\ref{lem:sample-target}. By Definition~\ref{def:Efficient classical learner under squared loss}, $\mathcal A$ runs in polynomial time and the hypothesis $g$ is evaluable in polynomial time. Since $L(n)=\operatorname{poly}(n)$, the reconstruction and divisibility check also take polynomial time.

We next analyze the correctness of $\mathcal B$. Recall the definition of $z_j$ in Eq.~\eqref{eq:estimation-bit}. By Lemma~\ref{lem:sample-target},
\[
    F_n^{\cos,\pi}(x_{N,j,1-z_j})
    -
    F_n^{\cos,\pi}(x_{N,j,z_j})
    \ge
    2-4\nu(n),
\]
whereas
\[
    F_n^{\exp,\beta}(x_{N,j,1-z_j})
    -
    F_n^{\exp,\beta}(x_{N,j,z_j})
    \ge
    (1-e^{-\beta})(1-2\nu(n)).
\]
Since $\gamma_\beta<1$, $1-e^{-\beta}>\gamma_\beta$, and $\nu(n)$ is negligible, for all sufficiently large $n$, $F_n(x_{N,j,1-z_j})-F_n(x_{N,j,z_j})\ge\gamma_\beta$ for every $j\in[L(n)]$.
For a fixed challenge modulus $N$, define the squared loss conditioned on $N$ by
\begin{equation}
    R_N(g):=
    \mathbb{E}_{j\in[L(n)],\,b\in\{0,1\}}
    \left[
        \bigl(g(x_{N,j,b})-F_n(x_{N,j,b})\bigr)^2
    \right],
    \label{eq:conditional-risk}
\end{equation}
where $j$ and $b$ are chosen independently and uniformly at random. By construction, under the common input marginal $\mu_n$, $N\sim\mathcal D_n$, while $j\in[L(n)]$ and $b\in\{0,1\}$ are chosen independently and uniformly at random. Hence, $R_{\mu_n}(g)=\mathbb{E}_{N\sim\mathcal{D}_n}\left[R_N(g)\right]$. 
Suppose that $R_N(g)\le\tau(n)$. By Eq.~\eqref{eq:conditional-risk},
\[
    R_N(g)
    =
    \frac{1}{2L(n)}
    \sum_{j\in[L(n)]}
    \sum_{b\in\{0,1\}}
    \left(
        g(x_{N,j,b})-F_n(x_{N,j,b})
    \right)^2.
\]
Since all $2L(n)$ summands are nonnegative, for every $j\in[L(n)]$ and $b\in\{0,1\}$,
\[
    \left|
        g(x_{N,j,b})-F_n(x_{N,j,b})
    \right|
    \leq
    \sqrt{2L(n)\tau(n)}
    =
    \frac{\gamma_\beta}{8}.
\]
Together with the target-separation bound above, this gives
\[
\begin{aligned}
    g(x_{N,j,1-z_j})-g(x_{N,j,z_j})
    &\ge
    F_n(x_{N,j,1-z_j})-F_n(x_{N,j,z_j})
    -
    \left|
        g(x_{N,j,1-z_j})-F_n(x_{N,j,1-z_j})
    \right|
    -
    \left|
        g(x_{N,j,z_j})-F_n(x_{N,j,z_j})
    \right|
    \\
    &\ge
    \gamma_\beta-\frac{\gamma_\beta}{4}
    =
    \frac{3\gamma_\beta}{4}
    >
    0.
\end{aligned}
\]
Hence, $g(x_{N,j,z_j})<g(x_{N,j,1-z_j})$, so the comparison rule in Step~3 gives $\widehat z_j=z_j=\operatorname{bit}^{(L(n))}_j(\min\{p,q\})$ for every $j\in[L(n)]$.

It remains to show that $R_N(g)\le \tau(n)$ holds with sufficiently high probability. By Definition~\ref{def:Efficient classical learner under squared loss}, with probability at least $1-\delta(n)$ over $S$ and the randomness of $\mathcal A$, $R_{\mu_n}(g) \le \varepsilon(n)$.
Since $g$ is generated independently of the challenge modulus $N$, for any fixed realization of $g$ satisfying this bound, Markov's inequality gives
\[
    \Pr_{N\sim \mathcal D_n}
    [R_N(g)>\tau(n)]
    \le
    \frac{\varepsilon(n)}{\tau(n)}= \frac{1}{200}.
\]
Therefore, over the challenge modulus $N$, the training sample $S$, and the internal randomness of $\mathcal A$,
\[
    \Pr_{N\sim \mathcal D_n,S,\mathcal A}
    [R_N(g)>\tau(n)]
    \le
    \delta(n)
    +
    \frac{\varepsilon(n)}{\tau(n)}
    =
    \frac{1}{100}.
\]
Since $\mathcal D_n$ is the marginal of $\widetilde{\mathcal D}_n$ on $N$, the same probability bound holds when $(N,p,q)$ is drawn from $\widetilde{\mathcal D}_n$.
Whenever $R_N(g)\le\tau(n)$, we have $\widehat p=\min\{p,q\}$.
Consequently,
\[
    \Pr_{(N,p,q)\sim\widetilde {\mathcal D}_n,S,\mathcal A}
    [\widehat p=\min\{p,q\}]
    \ge
    \frac{99}{100}
    >
    \frac{2}{3}.
\]
Whenever $\widehat p=\min\{p,q\}$, the final divisibility check succeeds and $\mathcal B$ outputs
$
    \left(
        \widehat p,
        \frac{N}{\widehat p}
    \right)
    =
    (\min\{p,q\},\max\{p,q\}).
$
Thus, for all sufficiently large $n$,
\[
\Pr_{\substack{(N,p,q)\sim\widetilde{\mathcal{D}}_n,\, \mathcal{B}}}
\left[
\exists\,p',q'\in\mathbb{N}:
\mathcal{B}(1^n,N)=(p',q'),\
1<p',q'<N,\
p'q'=N
\right]
\ge \frac{99}{100}
> \frac{2}{3}.
\]
This proves Theorem~\ref{thm:main}.

\end{proof}

\begin{remark}\label{rem:generalization}
As a final remark, we note that the reduction in Theorem~\ref{thm:main} is not specific to integer factorization and actually applies more generally than stated in Theorem~\ref{thm:main}.
The same argument works whenever there is a polynomial-length witness $w(u)$ uniquely determined by an input $u$ such that ($i$) pairs $(u,w(u))$ are samplable in classical polynomial time, ($ii$) $w(u)$ can be recovered from $u$ with bounded error in quantum polynomial time and a candidate witness can be verified in classical polynomial time, and ($iii$) recovering $w(u)$ from $u$ drawn from the marginal distribution induced by the sampler in ($i$) is hard on average for classical randomized polynomial-time algorithms. By standard amplification, the recovery error in ($ii$) can be made negligible. In particular, this applies to injective one-way functions secure against classical randomized polynomial-time algorithms that are invertible in quantum polynomial time.
\end{remark}